\documentclass[11pt]{article}
\usepackage{graphicx}
\usepackage{tabularx}
\usepackage{url}
\usepackage{amsthm}
\usepackage{amsmath}
\usepackage{amsfonts}
\usepackage[labelfont=bf]{caption}
\usepackage{amssymb}

\theoremstyle{plain}
\newtheorem{theorem}{Theorem}
\newtheorem{lemma}{Lemma}
\newtheorem{proposition}{Proposition}

\newtheorem{corollary}{Corollary}
\theoremstyle{definition}
\newtheorem{definition}{Definition}

\theoremstyle{remark}
\newtheorem{remark}{Remark}

\date{}

\begin{document}

\title{Algorithmic Problems for Computation Trees}

\author{Mikhail Moshkov \\
Computer, Electrical and Mathematical Sciences \& Engineering Division \\
King Abdullah University of Science and Technology (KAUST) \\
Thuwal 23955-6900, Saudi Arabia\\ mikhail.moshkov@kaust.edu.sa
}

\maketitle

\begin{abstract}
In this paper, we study three algorithmic problems involving computation trees: the optimization, solvability, and satisfiability problems. The solvability problem is concerned with recognizing computation trees that solve problems. The satisfiability problem is concerned with recognizing sentences that are true in at least one structure from a given set of structures. We study how the decidability of the optimization problem depends on the decidability of the solvability and
satisfiability problems. We also consider various examples with both decidable and undecidable
solvability and satisfiability problems.
\end{abstract}

{\it Keywords}: structure, computation tree, optimization.

\section{Introduction\label{S5.0}}

This paper is devoted to the study of  three algorithmic problems related to computation trees: problems of optimization, solvability and satisfiability. Computation trees  can be used as algorithms for solving  problems of combinatorial optimization, computation geometry, etc. Computation trees are a natural generalization of decision trees: besides one-place operations of predicate type (attributes), which are used in decision trees, in computation trees
many-place predicate and function operations may be used. 
Algebraic computation trees and some their generalization have been studied most intensively \cite{Ben-Or83,Gabrielov17,Grigoriev96}. This paper continues the line of research introduced by the work \cite{Moshkov22a}: we do not concentrate on specific types of computation trees, such as algebraic, but study computation trees over arbitrary structures.

Our consideration is based on the notions of signature and structure of this signature. Signature $\sigma$ is a finite or countable set of predicate and function symbols with their arity. Structure $U$ of the signature $\sigma$ is a pair $(A,I)$, where $A$ is a nonempty set called the universe of $U$ and $I$ is an interpretation function mapping the symbols of $\sigma$ to predicates and functions in $A$.

Each computation tree over structure $U$ is a pair $(S,U)$ of a computation tree scheme
$S$ of the signature $\sigma$ and the structure $U$. Such computation trees are used to solve problems over structure $U$, each of
which is a pair $(s,U)$ of a problem scheme $s$ of the signature $\sigma$  and the structure $U$. The
considered notions are discussed in Sects. \ref{S5.1}-\ref{S5.3}.

For a given class $C$ of structures of the signature $\sigma$, we say that a scheme of computation tree $S$
solves a scheme of problem $s$ relative to the class $C$ if, for any
structure $U\in C$, the computation tree $(S,U)$ solves the problem $(s,U)$.

We define the notion of so-called strictly limited complexity measure $\psi $ over the signature $\sigma$
and consider the following problem of optimization: for given sentence $%
\alpha $ from a set of sentences $H$ of the signature $\sigma$ and scheme of problem $s$, we should
find a scheme of computation tree $S$, which satisfies the following conditions: (i) $%
S$ solves the scheme of problem $s$ relative to the class $C(\alpha )$
consisting of all structures $U\in C$ for which $\alpha $ is true in $U$ and
(ii) $S$ has the minimum $\psi $-complexity among such schemes of computation trees. We
study the decidability of this problem depending on the decidability of 
problems of solvability and satisfiability.

The problem of solvability: for arbitrary sentence $\alpha \in H$, scheme of
problem $s$, and scheme of computation tree $S$, we should recognize if the scheme of
computation tree $S$ solves the scheme of problem $s$ relative to the class $C(\alpha
)$.

In Sect. \ref{S5.4}, we prove that the problem of solvability is decidable
if and only if the following problem of satisfiability is decidable: for an
arbitrary sentence $\beta \in H\wedge \Phi _{\exists }$, we should recognize
if there exists a structure $U\in $ $C$ for which $\beta $ is true in $U$.
Here $H\wedge \Phi _{\exists }$ is the set of sentences $\alpha \wedge
\gamma $ such that $\alpha \in H$, $\gamma \in \Phi _{\exists }$, and $\Phi
_{\exists }$ is the set of sentences of the signature $\sigma$ in prenex form with the prefixes of the
form $\exists \cdots \exists $.

In the same section, we prove that if $H$ is equal to the set of all
sentences of the signature $\sigma$, then the problem of solvability is
decidable if and only if the theory of the class $C$ is decidable. This gives
us many examples of both decidability and undecidability of the satisfiability
problem.

We will find more examples in Sect. \ref{S5.5}, where we study the case of the satisfiability problem when   $H=H_{1}\wedge \cdots \wedge H_{n}$ and,
for $i=1,\ldots ,n$,  $H_{i}$ is the set of all sentences of the signature $\sigma$ in prenex
form with prefixes in a nonempty set of words in the alphabet $\{\forall
,\exists \}$ that is closed relative to subwords.
We describe all classes of sentences $%
H_{1}\wedge \cdots \wedge H_{n}\wedge \Phi _{\exists }$ for which the
problem of satisfiability is decidable both for the case, when $C$ is
the class of all structures of the signature $\sigma$, and for the case,
when $C$ is the class of all finite structures of the signature $\sigma$.

In Sect. \ref{S5.6}, we prove that, for any strictly limited complexity
measure, the problem of optimization is undecidable if the problem of
satisfiability is undecidable. We prove also that the problem of
optimization is decidable if the problem of solvability is decidable and the
considered strictly limited complexity measure satisfies some additional
condition.

\section{Sequences of Expressions\label{S5.1}}

In this section, we consider the notions of function and predicate
expressions of a signature $\sigma $ and the notion of an $n$-sequence of
such expressions. These notions are used, when we consider schemes of
problems and computation trees.

Let $\sigma $ be a finite or countable \emph{signature} containing generally
both predicate and function symbols and their arity. We will not consider
nullary predicate symbols. We will specifically mention when nullary
function symbols (constants) are not allowed.
\index{Signature}

Let $\omega =\{0,1,2,\ldots \}$, $E_{2}=$ $\{0,1\}$, and $X=\{x_{i}:i\in
\omega \}$ be the set of variables. For $n\in \omega \setminus \{0\}$, we
denote $X_{n}=\{x_{0},\ldots ,x_{n-1}\}$. Let $\alpha $ be a formula of the
signature $\sigma $. We will use the following notation: $\alpha ^{1}=\alpha
$ and $\alpha ^{0}=\lnot \alpha $.

\begin{definition}
A \emph{function expression} of the signature $\sigma $ is an expression of
the form $x_{j}\Leftarrow f(x_{l_{1}},\ldots ,x_{l_{m}})$, where $f$ is an $%
m $-ary function symbol of the signature $\sigma $. A \emph{predicate
expression} of the signature $\sigma $ is an expression of the form $%
x_{l_{1}}=x_{l_{2}}$ or an expression of the form $r(x_{l_{1}},\ldots
,x_{l_{k}})$, where $r$ is a $k$-ary predicate symbol of the signature $%
\sigma $.
\end{definition}
\index{Signature!function expression}
\index{Signature!predicate expression}

\begin{definition}
Let $n\in \omega \setminus \{0\}$. An $n$-\emph{sequence of expressions} is
a pair $(n,\beta )$, where $\beta $ is a finite sequence of function and
predicate expressions.
\end{definition}
\index{Signature!$n$-sequence of expressions}

Let $\beta =\beta _{1},\ldots ,\beta _{m}$. For $i=1,2,\ldots ,m$, we
correspond to the expression $\beta _{i}$ a sequence $M_{i}=t_{i0},t_{i1},%
\ldots $ of terms of the signature $\sigma $ with variables from $X_{n}$.
Let $M_{1}=x_{0},x_{1},\ldots ,x_{n-2},x_{n-1},x_{n-1},\ldots $. Let the
sequences $M_{1},\ldots ,M_{i}$, $m>i\geq 1$, be already defined. We now
define the sequence $M_{i+1}$ corresponding to the expression $\beta _{i+1}$%
. If $\beta _{i}$ is a predicate expression, then $M_{i+1}=M_{i}$. Let $%
\beta _{i}$ be a function expression $x_{j}\Leftarrow f(x_{l_{1}},\ldots
,x_{l_{m}})$. Then $M_{i+1}=t_{i0},\ldots ,t_{ij-1},f(t_{il_{1}},\ldots
,t_{il_{m}}),t_{ij+1},t_{ij+2},\ldots $.

We associate with each predicate expression $\beta _{i}$ of the sequence $%
\beta $ an atomic formula $\kappa (n,\beta ,\beta _{i})$ of the signature $%
\sigma $ with variables from $X_{n}$. Let $\beta _{i}$ be the expression $%
x_{l_{1}}=x_{l_{2}}$. Then $\kappa (n,\beta ,\beta _{i})$ is equal to $%
t_{il_{1}}=t_{il_{2}}$. Let $\beta _{i}$ be the expression $%
r(x_{l_{1}},\ldots ,x_{l_{k}})$. Then $\kappa (n,\beta ,\beta _{i})$ is
equal to $r(t_{il_{1}},\ldots ,t_{il_{k}})$.

\section{Schemes of Computation Trees and Computation Trees\label{S5.2}}

In this section, we discuss the notions of a scheme of computation tree and its
complete path, and the notion of a computation tree.

Later we will specifically clarify whether equality $=$ can be used in the
considered formulas of the signature $\sigma $.

For $n\in \omega \setminus \{0\}$, denote $Q_{n}^{=}(\sigma )$ the set of
formulas of the signature $\sigma $ with variables from $X_{n}$ that are
atomic formulas or negations of atomic formulas. Equality sign $=$ in $%
Q_{n}^{=}(\sigma )$ means that formulas of the form $t_{1}=t_{2}$ and $\lnot
(t_{1}=t_{2})$ belong to $Q_{n}^{=}(\sigma )$, where $t_{1}$ and $t_{2}$ are
terms of the signature $\sigma $ with variables from $X_{n}$. We denote by $%
Q_{n}(\sigma )$ the set of formulas from $Q_{n}^{=}(\sigma )$ that do not
contain equality. The set $Q_{n}(\sigma )$ contains only formulas of the
form $r(t_{1},\ldots ,t_{k})$ or $\lnot r(t_{1},\ldots ,t_{k})$, where $r$
is a $k$-ary predicate symbol from $\sigma $ and $t_{1},\ldots ,t_{k}$ are
terms of the signature $\sigma $ with variables from $X_{n}$.

\begin{definition}
A \emph{scheme of computation tree} of
the signature $\sigma $ is a pair $S=(n,G)$, where $n\in \omega \setminus
\{0\}$ and $G$ is a finite directed tree with root. This tree has three
types of nodes: function, predicate and terminal. A \emph{function} node is
labeled with a function expression of the signature $\sigma $. One edge
leaves the function node. This edge is not labeled. A \emph{predicate} node
is labeled with a predicate expression of the signature $\sigma $. This
expression can be of the form $x_{l_{1}}=x_{l_{2}}$ if we allow equality.
Two edges leave the predicate node. One edge is labeled with the number $0$
and another edge is labeled with the number $1$. A \emph{terminal} node is
labeled with a number from $\omega $. This node has no leaving edges.
\end{definition}
\index{Scheme of computation tree of signature}

The set $X_{n}$ will be called the \emph{set of input variables} of the
scheme of computation tree $S$. Usually, we will not distinguish the scheme $S$ and
its graph $G$.
\index{Scheme of computation tree of signature!set of input variables}

\begin{definition}
A \emph{complete path }of the scheme $S$ is a directed path, which starts in
the root and finishes in a terminal node of $S$.
\end{definition}
\index{Scheme of computation tree of signature!complete path}

Let $\tau =w_{1},d_{1},w_{2},d_{2},\ldots ,w_{m},d_{m},w_{m+1}$ be a complete
path of the scheme $S$. We denote $\beta =\beta _{1},\ldots ,\beta _{m}$ the
sequence of function and predicate expressions attached to nodes $%
w_{1},\ldots ,w_{m}$. Let us consider the $n$-sequence of expressions $%
(n,\beta )$. We correspond to each predicate node $w_{i}$ of the path $\tau $
the formula $\kappa (n,\beta ,\beta _{i})^{c}$ from $Q_{n}^{=}(\sigma )$,
where $c$ is the number attached to the edge $d_{i}$.

We now define the set of formulas $F(\tau )$ associated with the complete
path $\tau $. If $\tau $ contains predicate nodes, then $F(\tau )$ is the
set of formulas associated with predicate nodes of $\tau $. If $\tau $ does
not contain predicate nodes, then $F(\tau )=\{\mu _{\tau }\}$, where $\mu
_{\tau }$ is the formula $x_{0}=x_{0}$ if equality is allowed and $\lnot
r(x_{0},\ldots ,x_{0})\vee r(x_{0},\ldots ,x_{0})$ if equality is not
allowed and $r$ is a predicate symbol of the signature $\sigma $. Note that
if $\tau $ does not contain predicate nodes, then $\tau $ is the only
complete path of the scheme $S$. Let $F(\tau )=\{\alpha _{1},\ldots ,\alpha
_{k}\}$. Denote by $\pi _{\tau }(S)$ the formula $\alpha _{1}\wedge \cdots
\wedge \alpha _{k}$. Denote by $t_{\tau }$ the number attached to the
terminal node $w_{m+1}$ of the path $\tau $.

\begin{definition}
Let $U$ be a \emph{structure} of the signature $\sigma $ with the universe $%
A $ and $S=(n,G)$ be a scheme of computation tree of the signature $\sigma $.
The pair $\Gamma =(S,U)$ will be called a \emph{computation tree}  over $U$ with the
set of input values $X_{n}$. The scheme $S$ will be called the \emph{scheme }%
of the computation tree $\Gamma $.
\end{definition}
\index{Structure of signature}
\index{Computation tree over structure}
\index{Computation tree over structure!scheme}

\begin{definition}
A complete path $\tau $ of the scheme $S$ will be called \emph{realizable}
in the structure $U$ for a tuple $\bar{a}\in A^{n}$ if any formula from $%
F(\tau )$ is true in $U$ on the tuple $\bar{a}$.
\end{definition}
\index{Computation tree of signature!complete path!realizable}

One can show that there exists exactly one complete path of $S$ that is
realizable in the structure $U$ for the tuple $\bar{a}$.

We correspond to the computation tree $\Gamma $ a function $\varphi _{\Gamma
}:A^{n}\rightarrow \omega $. Let $\bar{a}\in A^{n}$ and $\tau $ be a
realizable in the structure $U$ for the tuple $\bar{a}$ complete path of the
scheme $S$. Then $\varphi _{\Gamma }(\bar{a})=t_{\tau }$. We will say that
the computation tree $\Gamma $ \emph{implements} the function $\varphi _{\Gamma }$.
\index{Computation tree over structure!implementing function}

We denote by $\mathrm{Tree}^{=}(\sigma )$ the set of schemes of
computation trees of the signature $\sigma $. Denote by $\mathrm{Tree}(\sigma )$ the
set of schemes from $\mathrm{Tree}^{=}(\sigma )$ that do not contain
equality in expressions attached to nodes.

\section{Schemes of Problems and Problems\label{S5.3}}

In this section, we consider the notions of a scheme of problem and its
special representation, and the notion of a problem. We also define the
notion of a computation tree solving a problem.

\begin{definition}
A \emph{scheme of problem} of the signature $\sigma $ is a tuple $s=(n,\nu
,\beta _{1},\ldots ,\beta _{m})$, where $n\in \omega \setminus \{0\}$, $m\in
\omega \setminus \{0\}$, $\beta _{1},\ldots ,\beta _{m}$ are function and
predicate expressions of the signature $\sigma $, and there is $k\in \omega
\setminus \{0\}$ such that $\nu :E_{2}^{k}\rightarrow \omega $ and there are
exactly $k$ predicate expressions in the sequence $\beta _{1},\ldots ,\beta
_{m}$.
\end{definition}
\index{Scheme of problem of signature}

The set $X_{n}$ is called the \emph{set of input variables }for the scheme
of problem $s$. We denote by $\beta $ the sequence of expressions $\beta
_{1},\ldots ,\beta _{m}$. Let us consider the $n$-sequence of expressions $%
(n,\beta )$. Let $\beta _{i_{1}},\ldots ,\beta _{i_{k}}$ be all predicate
expressions in the sequence $\beta $. We correspond to each predicate
expression $\beta _{i_{j}}$ the formula $\alpha _{j}=\kappa (n,\beta ,\beta
_{i_{j}})$ from $Q_{n}^{=}(\sigma )$. The tuple $(n,\nu ,\alpha _{1},\ldots
,\alpha _{k})$ will be called the \emph{special representation} of the
scheme of problem $s$. We correspond to the scheme $s$ and to each tuple $%
\bar{\delta}=(\delta _{1},\ldots ,\delta _{k})\in $ $E_{2}^{k}$ the formula $%
\pi _{\bar{\delta}}(s)=\alpha _{1}^{\delta _{1}}\wedge \cdots \wedge \alpha
_{k}^{\delta _{k}}$.
\index{Scheme of problem of signature!set of input variables}
\index{Scheme of problem of signature!special representation}

\begin{definition}
Let $U$ be a \emph{structure} of the signature $\sigma $ consisting of a
universe $A$ and an interpretation $I$ of symbols from $\sigma $. Let $s=(n,\nu
,\beta _{1},\ldots ,\beta _{m})$ be a scheme of problem of the signature $%
\sigma $. The pair $z=(s,U)$ will be called a \emph{problem} over $U$ with
the set of input variables $X_{n}$. The scheme $s$ will be called the \emph{%
scheme }of the problem $z$.
\end{definition}
\index{Problem over structure}
\index{Problem over structure!scheme}

Let us consider the special representation $(n,\nu ,\alpha _{1},\ldots
,\alpha _{k})$ of the scheme of problem $s$. The interpretation $I$
corresponds to the formula $\alpha _{j}$, $j=1,\ldots ,k$, a function from $%
A^{n}$ to $E_{2}$ in a natural way: any predicate including predicates of
the form $x_{l_{1}}=x_{l_{2}}$ can be considered as a function with values
from the set $E_{2}$.

We correspond to the problem $z$ the function $\varphi _{z}:A^{n}\rightarrow
\omega $ defined as follows: $\varphi _{z}(\bar{a})=\nu (\alpha _{1}(\bar{a}%
),\ldots ,\alpha _{k}(\bar{a}))$ for any $\bar{a}\in A^{n}$. The problem $z$
may be interpreted as a problem of searching for the number $\varphi _{z}(%
\bar{a})$ for a given $\bar{a}\in A^{n}$.

\begin{definition}
Let $S=(n,G)$ be a scheme of computation tree of the signature $\sigma $. We will say
that the computation tree $\Gamma =(S,U)$ \emph{solves} the problem $z=(s,U)$ if the
functions $\varphi _{\Gamma }$ and $\varphi _{z}$ coincides.
\end{definition}

We denote by $\mathrm{Probl}^{=}(\sigma )$ the set of schemes of problems of
the signature $\sigma $. Denote by $\mathrm{Probl}(\sigma )$ the set of
schemes from $\mathrm{Probl}^{=}(\sigma )$ that do not contain equality in
predicate expressions.

\section{Solvability Versus Satisfiability\label{S5.4}}

In this section, we consider two algorithmic problems. The problem of
solvability is related to the recognition of schemes of computation trees solving
schemes of problems relative to a class of structures. The problem of
satisfiability is about the recognition of sentences that are true in at least one structure from a given set of structures. We show that the problem of solvability is decidable if and only
if corresponding to it   problem of satisfiability is decidable. We also consider various
examples with both decidable and undecidable solvability and satisfiability problems.

Let $C$ be a class of structures of the signature $\sigma $.

\begin{definition}
We will say that a scheme of computation tree $S=(n_{1},G)$ from $\mathrm{Tree}%
^{=}(\sigma )$ \emph{solves} a scheme of problem $s=(n_{2},\nu ,\beta
_{1},\ldots ,\beta _{m})$ from $\mathrm{Probl}^{=}(\sigma )$ relative to the
class $C$, if $n_{1}=n_{2}$ and either $C=\emptyset $ or $C\neq \emptyset $
and, for any structure $U$ from $C$, the computation tree $(S,U)$ solves the problem $%
(s,U)$.
\end{definition}

For a sentence (formula without free variables) $\alpha $ of the signature $%
\sigma $, denote $C(\alpha )=\{U:U\in $ $C,U\models \alpha \}$, where the
notation $U\models \alpha $ means that the sentence $\alpha $ is true in $U$%
. Let $H$ be a nonempty set of sentences of the signature $\sigma $ and $$(%
\mathrm{Probl},\mathrm{Tree}\mathcal{)\in }\{(\mathrm{Probl}^{=}(\sigma ),%
\mathrm{Tree}^{=}(\sigma )),(\mathrm{Probl}(\sigma ),\mathrm{Tree}(\sigma
))\}.$$ We can consider the pair $(\mathrm{Probl}(\sigma ),\mathrm{Tree}%
(\sigma ))$ only if the signature $\sigma $ contains predicate symbols.

\begin{definition}
We now define the \emph{problem of solvability} for the quadruple $(\mathrm{%
Probl},\mathrm{Tree},H,$ $C)$: for arbitrary $\alpha \in H$, $s\in \mathrm{%
Probl}$, and $S\in $ $\mathrm{Tree}$, we should recognize if the scheme of
computation tree $S$ solves the scheme of problem $s$ relative to the class $C(\alpha
)$.
\end{definition}
\index{Algorithmic problem!problem of solvability}

\begin{definition}
Let us define the \emph{problem of satisfiability} for the pair $(H,C)$: for
an arbitrary $\alpha \in H$, we should recognize if there exists a structure
$U\in $ $C$ such that $U\models \alpha $.
\end{definition}
\index{Algorithmic problem!problem of satisfiability}

This is the general definition of the satisfiability problem. Later, for each version of the solvability problem, we will consider its corresponding version of the satisfiability problem and prove that they are both decidable or both undecidable.

Let $H_{1},\ldots ,H_{m}$ be nonempty sets of sentences of the signature $%
\sigma $. We denote by $H_{1}\wedge \cdots \wedge H_{m}$ the class of
sentences of the form $\alpha _{1}\wedge \cdots \wedge \alpha _{m}$, where $%
\alpha _{1}\in H_{1},\ldots ,\alpha _{m}\in H_{m}$. Let $\Pi $ be a set of
words in the alphabet $\{\forall ,\exists \}$. Denote by $\Phi ^{=}(\Pi
,\sigma )$ the class of all sentences in prenex form of the signature $%
\sigma $ with prefixes from $\Pi $. Denote by $\Phi (\Pi ,\sigma )$ the
class of all sentences from $\Phi ^{=}(\Pi ,\sigma )$ that do not contain
equality.

Let $w$ be a word in the alphabet $\{\forall ,\exists ,\forall ^{\ast
},\exists ^{\ast }\}$. The set $P(w)$ of prefixes is defined in the
following way:
\begin{eqnarray*}
P(\forall ^{n}) =\{\forall ^{i}:0\leq i\leq n\},\;P(\exists ^{n})=\{\exists
^{i}:0\leq i\leq n\}, \\
P(\forall ^{\ast }) =\{\forall ^{i}:i\in \omega \},\;P(\exists ^{\ast
})=\{\exists ^{i}:i\in \omega \}, \\
P(w_{1}w_{2}) =\{u_{1}u_{2}:u_{1}\in P(w_{1}),u_{2}\in P(w_{2})\}.
\end{eqnarray*}

Denote $H^{=}(\sigma )$ the set of sentences of the signature $\sigma $.
Denote $H(\sigma )$ the set of sentences of the signature $\sigma $ that do
not contain equality.

\begin{theorem}
\label{T5.1}Let $C$ be a nonempty class of structures of the signature $%
\sigma $ and $H$ be a nonempty set of sentences of the signature $\sigma $.

{\rm (a) }Let $H\subseteq H^{=}(\sigma )$. The problem of solvability for
the quadruple $(\mathrm{Probl}^{=}(\sigma ),\mathrm{Tree}^{=}(\sigma ),$ $H,C)$
is decidable if and only if the problem of satisfiability for the pair $%
(H\wedge \Phi ^{=}(P(\exists ^{\ast }),\sigma ),C)$ is decidable.

{\rm (b) }Let $H\subseteq H(\sigma )$. The problem of solvability for the
quadruple $(\mathrm{Probl}(\sigma ),\mathrm{Tree}(\sigma ),H,C)$ is
decidable if and only if the problem of satisfiability for the pair $%
(H\wedge \Phi (P(\exists ^{\ast }),\sigma ),C)$ is decidable.
\end{theorem}

\begin{proof}
(a) Let the problem of solvability for the quadruple $(\mathrm{Probl}%
^{=}(\sigma ),\mathrm{Tree}^{=}(\sigma ),H,C)$ be decidable. We now
describe an algorithm for solving the problem of satisfiability for the pair
$(H\wedge \Phi ^{=}(P(\exists ^{\ast }),\sigma ),C)$. Let $\alpha \in H$ and
$\beta \in \Phi ^{=}(P(\exists ^{\ast }),\sigma )$. We construct a sentence
of the signature $\sigma $ that is logically equivalent to the sentence $%
\beta $ and is of the form $\exists x_{0}\ldots \exists x_{n-1}((\beta
_{11}\wedge \cdots \wedge \beta _{1m_{1}})\vee \cdots \vee (\beta
_{k1}\wedge \cdots \wedge \beta _{km_{k}}))$, where, for $j=1,\ldots ,k$ and
$i=1,\ldots ,m_{j}$, $\beta _{ji}\in Q_{n}^{=}(\sigma )$.

For $j=1,\ldots ,k$, we construct a scheme of computation tree $S_{j}=(n,G_{j})$ from
$\mathrm{Tree}^{=}(\sigma )$ such that, for any structure $U$ of the
signature $\sigma $ and any tuple $\bar{a}\in A^{n}$, where $A$ is the
universe of $U$,%
\[
\varphi _{(S_{j},U)}(\bar{a})=\left\{
\begin{array}{ll}
0, & \mathrm{if}\ U\models \lnot (\beta _{j1}(\bar{a})\wedge \cdots \wedge
\beta _{jm_{j}}(\bar{a})), \\
1, & \mathrm{if}\ U\models \beta _{j1}(\bar{a})\wedge \cdots \wedge \beta
_{jm_{j}}(\bar{a}).%
\end{array}%
\right.
\]%
We denote by $s_{0}=(n,\nu ,x_{0}=x_{0})$ the scheme of problem from $%
\mathrm{Probl}^{=}(\sigma )$ such that $\nu :E_{2}\rightarrow \{0\}$, i.e., $%
\varphi _{(s_{0},U)}\equiv 0$ for any structure $U$ of the signature $\sigma
$.

Using an algorithm that solves the problem of solvability for the quadruple $%
(\mathrm{Probl}^{=}(\sigma ),$ $\mathrm{Tree}^{=}(\sigma ),H,C)$, for $%
j=1,\ldots ,k$, we recognize if the scheme of computation tree $S_{j}$ solves the
scheme of problem $s_{0}$ relative to the class $C(\alpha )$. It is not
difficult to show that a structure $U\in $ $C$ such that $U\models \alpha
\wedge \beta $ exists if and only if there exists $j\in \{1,\ldots ,k\}$ for
which the scheme of computation tree $S_{j}$ does not solve the scheme of problem $%
s_{0}$ relative to the class $C(\alpha )$.

Let the problem of satisfiability for the pair $(H\wedge \Phi ^{=}(P(\exists
^{\ast }),\sigma ),C)$ be decidable. We now describe an algorithm solving
the problem of solvability for the quadruple $(\mathrm{Probl}^{=}(\sigma ),$ $
\mathrm{Tree}^{=}(\sigma ),H,C)$. Let $\alpha \in H$, $S=(n_{1},G)$ be a
scheme of computation tree from $\mathrm{Tree}^{=}(\sigma )$, and $s=(n_{2},\nu
,\beta _{1},\ldots ,\beta _{m})$ be a scheme of problem from $\mathrm{Probl}%
^{=}(\sigma )$ with the special representation $(n_{2},\nu ,\alpha
_{1},\ldots ,\alpha _{k})$.

If $n_{1}\neq n_{2}$, then the scheme of computation tree $S$ does not solve the
scheme of problem $s$ relative to the class $C(\alpha )$. Let $n_{1}=n_{2}$.
Let $\tau $ be a complete path of the scheme of computation tree $S$ and $\bar{\delta}%
\in $ $E_{2}^{k}$. Denote $\gamma _{\tau \bar{\delta}}=$ $\exists
x_{0}\ldots \exists x_{n_{1}-1}(\pi _{\tau }(S)\wedge \pi _{\bar{\delta}%
}(s)) $.

Using an algorithm solving the problem of satisfiability for the pair $%
(H\wedge \Phi ^{=}(P(\exists ^{\ast }),\sigma ),C)$, for each pair of path $%
\tau $ and tuple $\bar{\delta}$ such that $t_{\tau }\neq \nu (\bar{\delta})$%
, we check if there exists a structure $U\in $ $C$ in which the sentence $%
\alpha \wedge \gamma _{\tau \bar{\delta}}$ is true. One can show that the
scheme of computation tree $S$ does not solve the scheme of problem $s$ relative to
the class $C(\alpha )$ if and only if, for at least one of the considered
pairs $\tau $ and $\bar{\delta}$, the sentence $\alpha \wedge \gamma _{\tau
\bar{\delta}}$ is true in a structure from the class $C$.

(b) The second part of the theorem statement can be proved in a similar way.
\end{proof}

We now consider some corollaries of Theorem \ref{T5.1}.

\begin{definition}
Let $C$ be a nonempty class of structures of the signature $\sigma $. The
\emph{theory} $\mathrm{Th}(C)$ of the class $C$ consists of all sentences
from $H^{=}(\sigma )$ that are true in all structures from $C$. This theory
is called \emph{decidable} if there exists an algorithm that, for a given
sentence from $H^{=}(\sigma )$, recognizes if this sentence belongs to $%
\mathrm{Th}(C)$.
\end{definition}
\index{Theory!of class of structures}

\begin{proposition}
\label{P5.1}Let $C$ be a nonempty class of structures of the signature $%
\sigma $. The problem of solvability for the quadruple $(\mathrm{Probl}%
^{=}(\sigma ),\mathrm{Tree}^{=}(\sigma ),H^{=}(\sigma ),C)$ is decidable if
and only if the theory $\mathrm{Th}(C)$ of the class $C$ is decidable.
\end{proposition}

\begin{proof}
Let the theory $\mathrm{Th}(C)$ be decidable. Let $\alpha \in H^{=}(\sigma
)\wedge \Phi ^{=}(P(\exists ^{\ast }),\sigma )$. It is clear that the
sentence $\alpha $ is true in a structure from $C$ if and only if the
sentence $\lnot \alpha $ does not belong to the theory $\mathrm{Th}(C)$.
Taking into account that the theory $\mathrm{Th}(C)$ is decidable, we obtain
that the problem of satisfiability for the pair $(H^{=}(\sigma )\wedge \Phi
^{=}(P(\exists ^{\ast }),\sigma ),C)$ is decidable. Using Theorem \ref{T5.1}%
, we obtain that the problem of solvability for the quadruple $(\mathrm{Probl%
}^{=}(\sigma ),\mathrm{Tree}^{=}(\sigma ),H^{=}(\sigma ),C)$ is decidable.

Let the problem of solvability for the quadruple $(\mathrm{Probl}^{=}(\sigma
),\mathrm{Tree}^{=}(\sigma ),H^{=}(\sigma ),C)$ be decidable. Let $\alpha
\in H^{=}(\sigma )$. Evidently, $\alpha $ does not belong to the theory $%
\mathrm{Th}(C)$ if and only if the sentence $\lnot \alpha \wedge \exists
x_{0}(x_{0}=x_{0})$ from $H^{=}(\sigma )\wedge \Phi ^{=}(P(\exists ^{\ast
}),\sigma )$ is true in a structure from $C$. Using Theorem \ref{T5.1}, we
obtain that the problem of satisfiability for the pair $(H^{=}(\sigma
)\wedge \Phi ^{=}(P(\exists ^{\ast }),\sigma ),C)$ is decidable. Hence the
theory $\mathrm{Th}(C)$ is decidable.
\end{proof}

We now consider examples of classes of structures with decidable theory \cite%
{Ax68,ErshovLTT65}:

\begin{itemize}
\item The class of abelian groups and the class of finite abelian groups.

\item The class of finite fields.

\item The field of real numbers.

\item The field of complex numbers.

\item Addition of natural numbers.
\end{itemize}

We now consider examples of classes of structures with undecidable theory
\cite{ErshovLTT65,ErshovP11}:

\begin{itemize}
\item The class of groups and the class of finite groups.

\item The class of fields.

\item Arithmetic of natural numbers.
\end{itemize}

\section{Decidable Classes of Sentences\label{S5.5}}
In this section, we study decidable classes of sentences of a special kind. We first describe the details of our study.

Let $\sigma $ be a finite or countable signature that does not contain
nullary predicate symbols and nullary function symbols. A set $\Pi $ of
words in the alphabet $\{\forall ,\exists \}$ will be called \emph{closed
relative to subwords} if, for any word $w\in \Pi $, any word obtained from $%
w$ by the removal of some letters, belongs to $\Pi $. Let $\Pi _{1},\ldots
,\Pi _{n}$ be nonempty sets of words in the alphabet $\{\forall ,\exists \}$
closed relative to subwords. Let $\Pi _{n+1}=P(\exists ^{\ast })$. Denote $%
L^{=}=\Phi ^{=}(\Pi _{1},\sigma )\wedge \cdots \wedge \Phi ^{=}(\Pi
_{n},\sigma )$, $K^{=}=L^{=}\wedge \Phi ^{=}(\Pi _{n+1},\sigma )$, $L=\Phi
(\Pi _{1},\sigma )\wedge \cdots \wedge \Phi (\Pi _{n},\sigma )$, and $%
K=L\wedge \Phi (\Pi _{n+1},\sigma )$.

We denote by $C_{\sigma }$ the class of all structures of the signature $%
\sigma $ and by $C_{\sigma }^{fin}$ the class of all finite structures of
the signature $\sigma $. The next statement follows directly from Theorem %
\ref{T5.1}. It describes the connections between the solvability problems under consideration and the corresponding satisfiability problems.

\begin{proposition}
\label{P5.2}{\rm (a)} The problem of solvability for the quadruple $(%
\mathrm{Probl}^{=}(\sigma ),\mathrm{Tree}^{=}(\sigma ),L^{=},$ $C_{\sigma })$
is decidable if and only if the problem of satisfiability for the pair $%
(K^{=},C_{\sigma })$ is decidable.

{\rm (b)} The problem of solvability for the quadruple $(\mathrm{Probl}%
^{=}(\sigma ),\mathrm{Tree}^{=}(\sigma ),L^{=},C_{\sigma }^{fin})$ is
decidable if and only if the problem of satisfiability for the pair $%
(K^{=},C_{\sigma }^{fin})$ is decidable.

{\rm (c)} The problem of solvability for the quadruple $(\mathrm{Probl}%
(\sigma ),\mathrm{Tree}(\sigma ),L,C_{\sigma })$ is decidable if and only
if the problem of satisfiability for the pair $(K,C_{\sigma })$ is decidable.

{\rm (d)} The problem of solvability for the quadruple $(\mathrm{Probl}%
(\sigma ),\mathrm{Tree}(\sigma ),L,C_{\sigma }^{fin})$ is decidable if and
only if the problem of satisfiability for the pair $(K,C_{\sigma }^{fin})$
is decidable.
\end{proposition}

Let $H$ be a class of sentences of the signature $\sigma $.

\begin{definition}
We will say that the class $H$ is \emph{decidable} if the problem of
satisfiability for the pair $(H,C_{\sigma })$ is decidable and the problem
of satisfiability for the pair $(H,C_{\sigma }^{fin})$ is decidable.
\end{definition}
\index{Class of sentences!decidable class}

\begin{definition}
We will say that a sentence $\alpha $ of the signature $\sigma $ is \emph{%
satisfiable} if $\alpha $ is true in a structure from $C_{\sigma }$. We will
say that a sentence $\alpha $ of the signature $\sigma $ is \emph{finite
satisfiable} if $\alpha $ is true in a structure from $C_{\sigma }^{fin}$.
\end{definition}

\begin{definition}
We will say that the class $H$ is a \emph{class of reduction} if there
exists an algorithm that corresponds to an arbitrary sentence $\alpha $ of
the signature $\sigma $ a sentence $\alpha ^{\prime }\in H$ such that $%
\alpha $ is satisfiable if and only if $\alpha ^{\prime }$ is satisfiable,
and $\alpha $ is finite satisfiable if and only if $\alpha ^{\prime }$ is
finite satisfiable. Note that each class of reduction is undecidable.
\end{definition}
\index{Class of sentences!class of reduction}

In this section, based on the results from \cite{BorgerGG97,Gurevich66,Gurevich69,Gurevich76}, we describe all decidable classes of sentences $K$ and $K^=$.

\subsection{Signatures with Predicate Symbols Only}

In this section, we consider a finite or countable signature $\sigma $ that
does not contain function symbols and nullary predicate symbols.

Denote $P_{1}=P(\exists ^{\ast }\forall ^{\ast })$, $P_{2}=P(\exists ^{\ast
}\forall ^{2}\exists ^{\ast })$, and $P_{3}=P(\exists ^{\ast }\forall ^{2})$%
. In \cite{Gurevich66}, a constant $r_{0}$ is defined in the following way: $%
r_{0}=\lfloor \log _{2}m^{\ast }\rfloor +13$, where $m^{\ast }$ is the
number of states in a universal Turing machine.

\begin{theorem}
\label{T5.2}Let the signature $\sigma $ contain at least $r_{0}+1$ predicate
symbols and do not contain function symbols. Let $\Pi _{1},\ldots ,\Pi _{n}$
be nonempty sets of words in the alphabet $\{\forall ,\exists \}$ closed
relative to subwords, $\Pi _{n+1}=P(\exists ^{\ast })$, $K^{=}=\Phi ^{=}(\Pi
_{1},\sigma )\wedge \cdots \wedge \Phi ^{=}(\Pi _{n},\sigma )\wedge \Phi
^{=}(\Pi _{n+1},\sigma )$, and $K=\Phi (\Pi _{1},\sigma )\wedge \cdots
\wedge \Phi (\Pi _{n},\sigma )\wedge \Phi (\Pi _{n+1},\sigma )$.

{\rm (a)} Either the class $K$ is decidable or it is a class of
reduction. The same applies to the class $K^{=}$.

{\rm (b) }$K$ is a class of reduction if and only if $K^{=}$ is a class
of reduction.

{\rm (c)} If the class $K^{=}$ is decidable, then, with the exception of
a finite number of sentences, the satisfiability for sentences from $K^{=}$
coincides with the finite satisfiability.

{\rm (d)} $K$ is decidable if and only if at least one of the following
four possibilities occurs:

\hspace{5pt} {\rm (d.1)} $\sigma $ contains only $1$-ary predicate symbols.

\hspace{5pt} {\rm (d.2)} For $i=1,\ldots ,n$, $\Pi _{i}\subseteq P_{1}$.

\hspace{5pt} {\rm (d.3)} For $i=1,\ldots ,n$, $\Pi _{i}\subseteq P_{2}$.

\hspace{5pt} {\rm (d.4)} There exists $i_{0}\in \{1,\ldots ,n\}$ such that $\Pi
_{i_{0}}\subseteq P_{1}\cup P_{2}$ and, for each $i\in \{1,\ldots
,n\}\setminus \{i_{0}\}$, $\Pi _{i}\subseteq P_{3}$.
\end{theorem}

\begin{proof}
Statements (a), (b) and (c) follow directly from Theorems 19 and 21 \cite%
{Gurevich66}.

Let $\sigma $ be an infinite signature. From Theorem 19 \cite{Gurevich66} it
follows that $K$ is a class of reduction if and only if the signature $%
\sigma $ contains at least one predicate symbol, which arity is at least
two, and at least one of the following two possibilities occurs:

\begin{itemize}
\item For some $i\in \{1,\ldots ,n+1\}$, $\Phi (\Pi _{i},\sigma )$ is a
reduction class.

\item For some $i,j\in \{1,\ldots ,n\}$ such that $i\neq j$, $\forall
\exists \in \Pi _{i}$ and $\forall ^{3}\in \Pi _{j}$.
\end{itemize}

From Theorem 3 \cite{Gurevich66} it follows that the class $\Phi (P_{1}\cup
P_{2},\sigma )$ is decidable. Therefore if, for the class $K$, at least one
of the possibilities (d.1), (d.2), (d.3), (d.4) occurs, then $K$ is
decidable.

Let $K$ be decidable. We now show that $K$ satisfies at least one of the
conditions (d.1), (d.2), (d.3), and (d.4).

If $\sigma $ contains only $1$-ary predicate symbols, then $K$ satisfies the
condition (d.1). Let $\sigma $ contain at least one predicate symbol $F$,
which arity is at least two. Then, for any $i\in \{1,\ldots ,n+1\}$, $\Phi
(\Pi _{i},\sigma )$ is a decidable class, and there are no $i,j\in
\{1,\ldots ,n\}$ such that $i\neq j$, $\forall \exists \in \Pi _{i}$ and $%
\forall ^{3}\in \Pi _{j}$. Taking into account that the signature $\sigma $
is infinite and contains the predicate symbol $F$, and using Theorem 3 \cite%
{Gurevich66}, we obtain that, for any $i\in \{1,\ldots ,n\}$, $\Pi
_{i}\subseteq P_{1}\cup P_{2}$.

A set of words $\Pi $ in the alphabet $\{\forall ,\exists \}$, which is
closed relative to subwords, will be called

\begin{itemize}
\item $\emptyset $\emph{-set} if $\Pi \subseteq P_{1}\cup P_{2}$, $\forall
\exists \notin \Pi $ and $\forall ^{3}\notin \Pi $.

\item $\forall ^{3}$\emph{-set} if $\Pi \subseteq P_{1}\cup P_{2}$, $\forall
\exists \notin \Pi $ and $\forall ^{3}\in \Pi $.

\item $\forall \exists $\emph{-set }if $\Pi \subseteq P_{1}\cup P_{2}$, $%
\forall \exists \in \Pi $ and $\forall ^{3}\notin \Pi $.

\item $\forall ^{3}$-$\forall \exists $\emph{-set} if $\Pi \subseteq
P_{1}\cup P_{2}$, $\forall \exists \in \Pi $ and $\forall ^{3}\in \Pi $.
\end{itemize}

Evidently, if $\Pi $ is a $\emptyset $-set, then $\Pi \subseteq P(\exists
^{\ast }\forall ^{2})=P_{3}$. If $\Pi $ is a $\forall ^{3}$-set, then $\Pi
\subseteq P_{1}$. If $\Pi $ is a $\forall \exists $-set, then $\Pi \subseteq
P_{2}$. If $\Pi $ is a $\forall ^{3}$-$\forall \exists $-set, then $\Pi
\subseteq P_{1}\cup P_{2}$. Note that $P_{3}=P_{1}\cap P_{2}$.

Let there exist $i_{0}$ such that $\Pi _{i_{0}}$ is a $\forall ^{3}$-$%
\forall \exists $-set. Then, for any $i\in \{1,\ldots ,n\}\setminus
\{i_{0}\} $, $\Pi _{i}$ is a $\emptyset $-set and $\Pi \subseteq P_{3}$.
Therefore $K$ satisfies the condition (d.4).

Let, for any $i\in \{1,\ldots ,n\}$, $\Pi _{i}$ be not a $\forall ^{3}$-$%
\forall \exists $-set.

If, for any $i\in \{1,\ldots ,n\}$, $\Pi _{i}$ is a $\emptyset $-set, then $%
K $ satisfies the conditions (d.2) and (d.3).

Let there exist $i\in \{1,\ldots ,n\}$ such that $\Pi _{i}$ is a $\forall
^{3}$-set. Then, evidently, for any $j\in \{1,\ldots ,n\}$, $\Pi _{j}$ is a $%
\emptyset $-set or a $\forall ^{3}$-set and $K$ satisfies the condition
(d.2).

Let there exist $i\in \{1,\ldots ,n\}$ such that $\Pi _{i}$ is a $\forall
\exists $-set. Then, evidently, for any $j\in \{1,\ldots ,n\}$, $\Pi _{j}$
is a $\emptyset $-set or a $\forall \exists $-set and $K$ satisfies the
condition (d.3).

Let $\sigma $ be a finite signature. Let the class $K$ satisfy at least one
of the conditions (d.1), (d.2), (d.3), (d.4). We add to $\sigma $ an
infinite set of $1$-ary predicate symbols. As a result, we obtain a
signature $\sigma _{1}$. It is clear that the class $K(\sigma _{1})=\Phi
(\Pi _{1},\sigma _{1})\wedge \cdots \wedge \Phi (\Pi _{n+1},\sigma _{1})$
satisfies at least one of the conditions (d.1), (d.2), (d.3), (d.4).
According to what was proven above, the class $K(\sigma _{1})$ is decidable.
Since $K\subseteq K(\sigma _{1})$, the class $K$ is decidable.

Let $\sigma $ contain at least $r_{0}+1$ predicate symbols. Let $K$ be
decidable. If $\sigma $ contains only $1$-ary predicate symbols, then $K$
satisfies the condition (d.1). Let $\sigma $ contain at least one predicate
symbol, which arity is greater than or equal to two. Taking into account
that $\Pi _{n+1}=P(\exists ^{\ast })$ and using Theorem 21 \cite{Gurevich66}%
, we obtain that, for any $i\in \{1,\ldots ,n\}$, $\forall \exists \forall
\notin \Pi _{i}\ $ and $\forall ^{3}\exists \notin \Pi _{i}$. Therefore, for
$i=1,\ldots ,n$, $\Pi _{i}\subseteq P_{1}\cup P_{2}$.

Taking into account that $\Pi _{n+1}=P(\exists ^{\ast })$ and using Theorem
21 \cite{Gurevich66}, we obtain that there are no $i,j\in \{1,\ldots ,n\}$
such that $i\neq j$, $\forall \exists \in \Pi _{i}$ and $\forall ^{3}\in \Pi
_{j}$.

Next, using arguments similar to the case of infinite signature $\sigma $,
where we considered $\emptyset $-sets, $\forall ^{3}$-sets, $\forall \exists
$-sets, and $\forall ^{3}$-$\forall \exists $-sets, we obtain that $K$
satisfies condition (d.2) or condition (d.3), or condition (d.4).
\end{proof}

\begin{remark}
If the signature $\sigma $ contains less than $r_{0}+1$ predicate symbols
and $K$ satisfies at least one of the conditions (d.1), (d.2), (d.3), (d.4),
then $K$ is decidable -- see the proof of the theorem.
\end{remark}

\subsection{Signatures with Function Symbols. Equality Is Not Allowed}

In this section, we consider a finite or countable signature $\sigma $ that
does not contain nullary predicate and function symbols and contains a
function symbol of arity greater than $0$. The equality $=$ can not be used
in the formulas of the signature $\sigma $.

\begin{theorem}
\label{T5.3}Let $\sigma $ be a signature with at least one function symbol
of arity greater than $0$. Let $\Pi _{1},\ldots ,\Pi _{n}$ be nonempty sets
of words in the alphabet $\{\forall ,\exists \}$ closed relative to
subwords, $\Pi _{n+1}=P(\exists ^{\ast })$, and $K=\Phi (\Pi _{1},\sigma
)\wedge \cdots \wedge \Phi (\Pi _{n},\sigma )\wedge \Phi (\Pi _{n+1},\sigma
) $.

{\rm (a)} Either the class $K$ is decidable or it is a class of reduction.

{\rm (b)} If the class $K$ is decidable, then for sentences from $K$, the
satisfiability coincides with the finite satisfiability.

{\rm (c)} The class $K$ is decidable if and only if at least one of the
following three conditions is satisfied:

\hspace{5pt} {\rm (c.1) }The signature $\sigma $ does not contain predicate symbols.

\hspace{5pt} {\rm (c.2)} The signature $\sigma $ contains only $1$-ary predicate and $%
1 $-ary function symbols.

\hspace{5pt} {\rm (c.3)} For $i=1,\ldots ,n$, $\Pi _{i}\subseteq P(\exists ^{\ast
}\forall \exists ^{\ast })$.
\end{theorem}

\begin{proof}
Let $K$ satisfy the condition (c.1). Then $K=\emptyset $ and, consequently, $%
K$ is decidable. Since $K=\emptyset $, the statement (b) holds for $K$.

Let $K$ satisfy the condition (c.2). By Theorem 7 \cite{Gurevich69}, the
class $\Phi (P,\sigma )$, where $P$ is the set of all words in the alphabet $%
\{\forall ,\exists \}$, is decidable. Evidently, there exists an algorithm
that transforms any sentence $\alpha \in K$ to a sentence $\alpha ^{\prime
}\in \Phi (P,\sigma )$ such that $\alpha ^{\prime }$ is satisfiable if and
only if $\alpha $ is satisfiable and $\alpha ^{\prime }$ is finite
satisfiable if and only if $\alpha $ is finite satisfiable. Therefore $K$ is
decidable. By Theorem 7 \cite{Gurevich69}, the statement (b) holds for $\Phi
(P,\sigma )$. Therefore the statement (b) holds for $K$.

Let $K$ satisfy the condition (c.3). Let $\alpha =\alpha _{1}\wedge \cdots
\wedge \alpha _{n+1}\in K$. By introducing some new nullary function symbols
$c_{1},\ldots ,c_{k}$ and changing variables, we can transform the sentence $%
\alpha $ to the form $\beta =\forall x_{0}\exists \ldots \exists u_{1}\wedge
\cdots \wedge \forall x_{0}\exists \ldots \exists u_{n+1}$ such that $%
u_{1},\ldots ,u_{n+1}$ are quantifier-free formulas of the signature $\sigma
^{\prime }=\sigma \cup \{c_{1},\ldots ,c_{k}\}$ with pairwise disjoint sets
of variables (except $x_{0}$) and, for any structure $U$ of the signature $%
\sigma $, $\alpha $ is true in $U$ if and only if $\beta $ is true in $U$
for some interpretation of new symbols $c_{1},\ldots ,c_{k}$ in the universe
of $U$. Denote $\gamma =\forall x_{0}(\exists \ldots \exists u_{1}\wedge
\cdots \wedge \exists \ldots \exists u_{n+1})$. Evidently, for any structure
$U^{\prime }$ of the signature $\sigma ^{\prime }$, $\gamma $ is true in $%
U^{\prime }$ if and only if $\beta $ is true in $U^{\prime }$. Further, we
can transform the sentence $\gamma $ to the form $\delta =\forall
x_{0}\exists \ldots \exists u$ such that $u$ is a quantifier-free formula of
the signature $\sigma ^{\prime }$ and, for any structure $U^{\prime }$ of
the signature $\sigma ^{\prime }$, $\delta $ is true in $U^{\prime }$ if and
only if $\gamma $ is true in $U^{\prime }$.

Let $y_{1},\ldots ,y_{k}$ be variables from $X$ that do not belong to $%
\delta $. Denote $\varepsilon =\exists y_{1}\ldots \exists y_{k}\forall
x_{0}\exists \ldots $ $\exists v$, where $v$ is obtained from $u$ by replacing
symbols $c_{1},\ldots ,c_{k}$ with variables $y_{1},\ldots ,y_{k}$,
respectively. One can show that, for any structure $U$ of the signature $%
\sigma $, $\varepsilon $ is true in $U$ if and only if $\delta $ is true in $%
U$ for some interpretation of new symbols $c_{1},\ldots ,c_{k}$ in the
universe of $U$. Therefore, for any structure $U$ of the signature $\sigma $%
, $\varepsilon $ is true in $U$ if and only if $\alpha $ is true in $U$. It
is clear that $\varepsilon \in \Phi (P(\exists ^{\ast }\forall \exists
^{\ast }),\sigma )$. Using Theorem 7 \cite{Gurevich69}, we obtain that the
class $\Phi (P(\exists ^{\ast }\forall \exists ^{\ast }),\sigma )$ is
decidable. Therefore the class $K$ is decidable.
By Theorem 7 \cite{Gurevich69}, the statement (b) holds for $\Phi (P(\exists ^{\ast }\forall \exists ^{\ast }),\sigma )$. Therefore the statement (b) holds for $K$.

Let $K$ do not satisfy any of the conditions (c.1), (c.2), and (c.3). Then $%
\sigma $ contains a predicate symbol $p$ and a function symbol $f$ such
that, for at least one of these symbols, its arity is at least $2$, and
there exists $i\in \{1,\ldots ,n\}$ such that $\Pi _{i}\nsubseteq P(\exists
^{\ast }\forall \exists ^{\ast })$. Let, for the definiteness, $i=1$. Since $%
\Pi _{1}\nsubseteq P(\exists ^{\ast }\forall \exists ^{\ast })$, the word $%
\forall ^{2}$ belongs to $\Pi _{1}$.

We now show that, for $i=2,\ldots ,n+1$, the set $\Phi (\Pi _{i},\sigma )$
contains a sentence that is true in any structure $U$ of the signature $%
\sigma $. Let $i\in \{2,\ldots ,n+1\}$ and $Q_{0}\ldots Q_{t}\in \Pi _{i}$,
where $Q_{0},\ldots ,Q_{t}\in \{\exists ,\forall \}$. Then the sentence $\pi
_{i}=Q_{0}x_{0}\ldots Q_{t}x_{t}(p(x_{0},\ldots ,x_{0})\vee \lnot
p(x_{0},\ldots ,x_{0}))$ is true in any structure $U$ of the signature $%
\sigma $ and $\pi _{i}\in \Phi (\Pi _{i},\sigma )$.

Let the arity of the symbol $p$ be greater than or equal to two. From
Theorem 7 \cite{Gurevich69} it follows that the class $\Phi (P(\forall
^{2}),\sigma _{1})$, where $\sigma _{1}$ consists of the 2-ary predicate
symbol $\rho $ and the 1-ary function symbol $\varphi $, is a class of
reduction. Let us show that there is an algorithm, which, for an arbitrary
sentence $\alpha \in \Phi (P(\forall ^{2}),\sigma _{1})$, constructs a
sentence $\alpha ^{\prime }\in K$ such that $\alpha ^{\prime }$ is
satisfiable if and only if $\alpha $ is satisfiable and $\alpha ^{\prime }$
is finite satisfiable if and only if $\alpha $ is finite satisfiable.

Let us replace in $\alpha $ every occurrence of $\rho (t_{1},t_{2})$ by $%
p(t_{1},t_{2},\ldots ,t_{2})$ and every occurrence of $\varphi (t)$ by $%
f(t,t,\ldots ,t)$, where $t_{1}$, $t_{2}$, and $t$ are arbitrary terms of
the signature $\sigma $. Denote by $\beta $ the obtained formula. Then as $%
\alpha ^{\prime }$ we can take the sentence $\beta \wedge \pi _{2}\wedge
\cdots \wedge \pi _{n+1}$, which belongs to $K$. Since $\Phi (P(\forall
^{2}),\sigma _{1})$ is a class of reduction, $K$ is also a class of
reduction.

Let the arity of the symbol $f$ be greater than or equal to two. From
Theorem 7 \cite{Gurevich69} it follows that the class $\Phi (P(\forall
^{2}),\sigma _{2})$, where $\sigma _{2}$ consists of the 1-ary predicate
symbol $\rho $ and the 2-ary function symbol $\varphi $, is a class of
reduction. Let us show that there is an algorithm, which, for an arbitrary
sentence $\alpha \in \Phi (P(\forall ^{2}),\sigma _{2})$, constructs a
sentence $\alpha ^{\prime }\in K$ such that $\alpha ^{\prime }$ is
satisfiable if and only if $\alpha $ is satisfiable and $\alpha ^{\prime }$
is finite satisfiable if and only if $\alpha $ is finite satisfiable.

Let us replace in $\alpha $ every occurrence of $\rho (t)$ by $p(t,t,\ldots
,t)$ and every occurrence of $\varphi (t_{1},t_{2})$ by $f(t_{1},t_{2},%
\ldots ,t_{2})$, where $t$, $t_{1}$, and $t_{2}$ are arbitrary terms of the
signature $\sigma $. Denote by $\beta $ the obtained formula. Then as $%
\alpha ^{\prime }$ we can take the sentence $\beta \wedge \pi _{2}\wedge
\cdots \wedge \pi _{n+1}$, which belongs to $K$. Since $\Phi (P(\forall
^{2}),\sigma _{2})$ is a class of reduction, $K$ is also a class of
reduction. This completes the proofs of statements (a), (b) and (c).
\end{proof}

\subsection{Signatures with Function Symbols. Equality Is Allowed%
}

\begin{theorem}
\label{T5.4}Let $\sigma $ be a signature with at least one function symbol.
Let $\Pi _{1},\ldots ,\Pi _{n}$ be nonempty sets of words in the alphabet $%
\{\forall ,\exists \}$ closed relative to subwords, $\Pi _{n+1}=P(\exists
^{\ast })$, and $K^{=}=\Phi ^{=}(\Pi _{1},\sigma )\wedge \cdots \wedge \Phi
^{=}(\Pi _{n},\sigma )\wedge \Phi ^{=}(\Pi _{n+1},\sigma )$.

{\rm (a)} Either the class $K^{=}$ is decidable or it is a class of
reduction.

{\rm (b)} The class $K^{=}$ is decidable if and only if at least one of
the following three conditions is satisfied:

\hspace{5pt} {\rm (b.1)} For $i=1,\ldots ,n$, $\Pi _{i}\subseteq P(\exists ^{\ast })$.

\hspace{5pt} {\rm (b.2)} The signature $\sigma $ contains only $1$-ary predicate
symbols, at most one $1$-ary function symbol, and does not contain function
symbols with arity greater than one.

\hspace{5pt} {\rm (b.3)} For $i=1,\ldots ,n$, $\Pi _{i}\subseteq P(\exists ^{\ast
}\forall \exists ^{\ast })$, the signature $\sigma $ contains at most one $1$%
-arity function symbol and does not contain function symbols with arity
greater than one. There are no any restrictions regarding predicate symbols.
\end{theorem}

\begin{proof}
Let $K^{=}$ satisfy the condition (b.1). Evidently, there is an algorithm,
which transforms an arbitrary sentence $\alpha \in K^{=}$ into a sentence $%
\alpha ^{\prime }\in \Phi ^{=}(P(\exists ^{\ast }),\sigma )$ such that $%
\alpha $ is satisfiable if and only if $\alpha ^{\prime }$ is satisfiable
and $\alpha $ is finite satisfiable if and only if $\alpha ^{\prime }$ is
finite satisfiable. Using Main Theorem \cite{Gurevich76}, we obtain that the
class $\Phi ^{=}(P(\exists ^{\ast }),\sigma )$ is decidable. Therefore the
class $K^{=}$ is decidable.

Let $K^{=}$ satisfy the condition (b.2). From Main Theorem \cite{Gurevich76}
it follows that the class $\Phi ^{=}(P,\sigma )$ is decidable, where $P$ is
the set of all words in the alphabet $\{\exists ,\forall \}$. Taking into
account that there is an algorithm that converts any sentence from $K^{=}$
to prenex form, we obtain that the class $K^{=}$ is decidable.

Let $K^{=}$ satisfy the condition (b.3). Now we will almost repeat the
reasoning from the proof of Theorem \ref{T5.3}. Let $\alpha =\alpha
_{1}\wedge \cdots \wedge \alpha _{n+1}\in K^{=}$. By introducing some new
nullary function symbols $c_{1},\ldots ,c_{k}$ and changing variables, we
can transform the sentence $\alpha $ to the form $\beta =\forall
x_{0}\exists \ldots \exists u_{1}\wedge \cdots \wedge \forall x_{0}\exists
\ldots \exists u_{n+1}$ such that $u_{1},\ldots ,u_{n+1}$ are
quantifier-free formulas of the signature $\sigma ^{\prime }=\sigma \cup
\{c_{1},\ldots ,c_{k}\}$ with pairwise disjoint sets of variables (except $%
x_{0}$) and, for any structure $U$ of the signature $\sigma $, $\alpha $ is
true in $U$ if and only if $\beta $ is true in $U$ for some interpretation
of new symbols $c_{1},\ldots ,c_{k}$ in the universe of $U$. Denote $\gamma
=\forall x_{0}(\exists \ldots \exists u_{1}\wedge \cdots \wedge \exists
\ldots \exists u_{n+1})$. Evidently, for any structure $U^{\prime }$ of the
signature $\sigma ^{\prime }$, $\gamma $ is true in $U^{\prime }$ if and
only if $\beta $ is true in $U^{\prime }$. Further, we can transform the
sentence $\gamma $ to the form $\delta =\forall x_{0}\exists \ldots \exists
u $ such that $u$ is a quantifier-free formula of the signature $\sigma
^{\prime }$ and, for any structure $U^{\prime }$ of the signature $\sigma
^{\prime }$, $\delta $ is true in $U^{\prime }$ if and only if $\gamma $ is
true in $U^{\prime }$.

Let $y_{1},\ldots ,y_{k}$ be variables from $X$ that do not belong to $%
\delta $. Denote $\varepsilon =\exists y_{1}\ldots \exists y_{k}\forall
x_{0}\exists \ldots $ $\exists v$, where $v$ is obtained from $u$ by replacing
symbols $c_{1},\ldots ,c_{k}$ with variables $y_{1},\ldots ,y_{k}$,
respectively. One can show that, for any structure $U$ of the signature $%
\sigma $, $\varepsilon $ is true in $U$ if and only if $\delta $ is true in $%
U$ for some interpretation of new symbols $c_{1},\ldots ,c_{k}$ in the
universe of $U$. Therefore, for any structure $U$ of the signature $\sigma $%
, $\varepsilon $ is true in $U$ if and only if $\alpha $ is true in $U$. It
is clear that $\varepsilon \in \Phi ^{=}(P(\exists ^{\ast }\forall \exists
^{\ast }),\sigma )$. Using Main Theorem \cite{Gurevich76}, we obtain that
the class $\Phi ^{=}(P(\exists ^{\ast }\forall \exists ^{\ast }),\sigma )$
is decidable. Therefore the class $K^{=}$ is decidable.

Let $K^{=}$ do not satisfy any of the conditions (b.1), (b.2), and (b.3)$\,$%
. One can show that in this case $K^{=}$ satisfy at least one of the
following three conditions:

(c.1) There exists $i\in \{1,\ldots ,n\}$ such that $\Pi _{i}\nsubseteq
P(\exists ^{\ast })$, i.e., $\forall \in \Pi _{i}$, and $\sigma $ contains
two $1$-ary function symbols.

(c.2) There exists $i\in \{1,\ldots ,n\}$ such that $\Pi _{i}\nsubseteq
P(\exists ^{\ast })$, i.e., $\forall \in \Pi _{i}$, and $\sigma $ contains a
function symbol with arity greater than one.

(c.3) There exists $i\in \{1,\ldots ,n\}$ such that $\Pi _{i}\nsubseteq
P(\exists ^{\ast }\forall \exists ^{\ast })$, i.e., $\forall ^{2}\in \Pi
_{i} $, and $\sigma $ contains a predicate symbol with arity greater than
one and a function symbol.

Denote by $\sigma _{1}$ the signature consisting of two $1$-ary function
symbols, $\sigma _{2}$ the signature consisting of one $2$-ary function
symbol, and $\sigma _{3}$ the signature consisting of one $2$-ary predicate
symbol and one $1$-ary function symbol. From Main Theorem \cite{Gurevich76}
(see also Theorem 4.0.1 \cite{BorgerGG97}) it follows that the classes $\Phi
^{=}(P(\forall ),\sigma _{1})$, $\Phi ^{=}(P(\forall ),\sigma _{2})$, and $%
\Phi ^{=}(P(\forall ^{2}),\sigma _{3})$ are classes of reduction.

If the class $K^{=}$ satisfies the condition (c.1), then $\Phi
^{=}(P(\forall ),\sigma _{1})$ can be reduced to $K^{=}$, and $K^{=}$ is a
class of reduction. If the class $K^{=}$ satisfies the condition (c.2), then
$\Phi ^{=}(P(\forall ),\sigma _{2})$ can be reduced to $K^{=}$, and $K^{=}$
is a class of reduction. If the class $K^{=}$ satisfies the condition (c.3),
then $\Phi ^{=}(P(\forall ^{2}),\sigma _{3})$ can be reduced to $K^{=}$, and
$K^{=}$ is a class of reduction. We will not go into details of the proof --
similar statements were considered in the proof of Theorem \ref{T5.3}.
\end{proof}

\section{Problem of Optimization\label{S5.6}}

In this section, we consider the problem of optimization of schemes of
computation trees and study how its decidability depends on the decidability of the
problems of solvability and satisfiability. We prove that, for any strictly
limited complexity measure, the problem of optimization is undecidable if
the problem of satisfiability is undecidable. We prove also that the problem
of optimization is decidable if the problem of solvability is decidable and
the considered strictly limited complexity measure satisfies some additional
condition. Note that the problem of solvability and corresponding to it problem
of satisfiability are either both decidable or both undecidable.

\subsection{Equality Is Not Allowed}

First, we consider the case, when the equality is not allowed.

Let $\sigma $ be a finite or countable signature. If $\sigma $ is finite,
then we represent it in the form $\sigma =\{q_{0},\ldots ,q_{m}\}$, where $%
q_{0},\ldots ,q_{m}$ are predicate and function symbols, each with its own
arity. If $\sigma $ is infinite, then we represent it in the form $\sigma
=\{q_{0},q_{1},\ldots \}$. We denote by $\sigma ^{\ast }$ the set of finite
words in the alphabet $\sigma $ including the empty word $\lambda $.

\begin{definition}
A \emph{complexity measure} over the signature $\sigma $ is an arbitrary map
$\psi :\sigma ^{\ast }\rightarrow \omega $. The complexity measure $\psi $
will be called \emph{strictly limited} if it is computable and, for any $%
\alpha _{1},\alpha _{2},\alpha _{3}\in \sigma ^{\ast }$, it possesses the
following property: if $\alpha _{2}\neq \lambda $, then $\psi (\alpha _{1}\alpha
_{2}\alpha _{3})>\psi (\alpha _{1}\alpha _{3})$.
\end{definition}
\index{Complexity measure over signature}
\index{Complexity measure over signature!strictly limited}

We now consider some examples of complexity measures. Let $\psi :\sigma
\rightarrow \omega \setminus \{0\}$. The function $\psi $ is called a \emph{%
weight function} for the signature $\sigma $. We extend the function $\psi $
to the set $\sigma ^{\ast }$ in the following way: $\psi (\alpha )=0$ if $%
\alpha =\lambda $ and $\psi (\alpha )=\sum_{j=1}^{m}\psi (q_{i_{j}})$ if $%
\alpha =q_{i_{1}}\cdots q_{i_{m}}$. This function is called a \emph{weighted
depth}. If $\psi (q_{i})=1$ for any $q_{i}$ $\in \sigma $, then the function
$\psi $ is called the \emph{depth} and is denoted by $h$. The depth and any
computable weighted depth are strictly limited complexity measures.
\index{Complexity measure over signature!weighted depth}
\index{Complexity measure over signature!depth}

Let $\psi $ be a complexity measure over the signature $\sigma $. We extend
it to the sets $\mathrm{Probl}(\sigma )\ $and $\mathrm{Tree}(\sigma )$. Let
$\beta $ be a finite sequence of function and predicate expressions of the
signature $\sigma $ that do not contain the equality. We correspond to $%
\beta $ a word $\mathrm{word}(\beta )$ from $\sigma ^{\ast }$. If the length
of $\beta $ is equal to $0$, then $\mathrm{word}(\beta )=\lambda $. If $%
\beta =\beta _{1},\ldots ,\beta _{m}$, then $\mathrm{word}(\beta )=$ $%
q_{i_{1}}\cdots q_{i_{m}}$, where, for $j=1,\ldots ,m$, $q_{i_{j}}$ is the
symbol of the signature $\sigma $ from the expression $\beta _{j}$.

Let $s=(n,\nu ,\beta _{1},\ldots ,\beta _{m})$ be a scheme of problem from
the set $\mathrm{Probl}(\sigma )$. Then $\psi (s)=\psi (\mathrm{word}(\beta
_{1},\ldots ,\beta _{m}))$.

Let $S=(n,G)$ be a scheme of computation tree from the set $\mathrm{Tree}(\sigma )$
and $\tau =w_{1},d_{1},w_{2},$ $d_{2},\ldots ,w_{m},d_{m},w_{m+1}$ be a complete
path of the scheme $S$. We denote $\beta _{\tau }=\beta _{1},\ldots ,\beta
_{m}$ the sequence of function and predicate expressions attached to nodes $%
w_{1},\ldots ,w_{m}$. Then $\psi (S)=\max \{\psi (\mathrm{word}(\beta _{\tau
})):\tau \in \mathrm{Path}(S)\}$, where $\mathrm{Path}(S)$ is the set of
complete paths of the scheme $S$. The value $\psi (S)$ will be called the $%
\psi $-\emph{complexity} of the scheme of computation tree $S$. We denote by $h(S)$
the depth of the scheme of computation tree $S$. By $\sigma (S)$, we denote the set
of symbols of the signature $\sigma $ used in function and predicate
expressions in the scheme $S$.
\index{Scheme of computation tree of signature!complexity}

\begin{lemma}
\label{L5.1} Let $\psi $ be a strictly limited complexity measure over
the signature $\sigma $ and $S$ be a scheme of computation tree from the set $\mathrm{%
Tree}(\sigma )$. Then the following statements hold:

{\rm (a)} $\psi (S)\geq h(S)$.

{\rm (b)} $\psi (S)\geq \max \{\psi (q_{j}):q_{j}\in \sigma (S)\}$.
\end{lemma}

\begin{proof}
Let $\alpha \in \sigma ^{\ast }$. Then it is easy to show that $\psi (%
\alpha )\geq |\alpha |$, where $|\alpha |$ is the length of the word $%
\alpha $, and $\psi (\alpha )\geq \psi (q_{j})$ for any letter $q_{j}$
in the word $\alpha $. Using these relations, one can show that the
statements of lemma hold.
\end{proof}

Let us remind that $H(\sigma )$ is the set of sentences of the signature $%
\sigma $ that do not contain equality. Let $\psi $ be a strictly limited
complexity measure over the signature $\sigma $, $C$ be a nonempty class of
structures of the signature $\sigma $, and $H$ be a nonempty subset of the
set $H(\sigma )$.

\begin{definition}
We now define the \emph{problem of optimization} for the triple $(\psi ,H,C)$%
: for arbitrary sentence $\alpha \in H$ and scheme of problem $s\in \mathrm{%
Probl}(\sigma )$, we should find a scheme of computation tree $S\in $ $\mathrm{Tree}%
(\sigma )$, which solves the scheme of problem $s$ relative to the class $%
C(\alpha )$ and has the minimum $\psi $-complexity. We will call such scheme
of computation tree \emph{optimal} relative to $\psi $, $s$, and $C(\alpha )$.
\end{definition}
\index{Algorithmic problem!problem of optimization}
\index{Scheme of computation tree of signature!optimal}

\begin{lemma}
\label{L5.2} The $\psi $-complexity of a scheme of computation tree that is optimal
relative to $\psi $, $s$, and $C(\alpha )$ is at most $\psi (s)$.
\end{lemma}

\begin{proof}
Let $s=(n,\nu ,\beta _{1},\ldots ,\beta _{m})$. It is easy to construct a
scheme of computation tree $S\in \mathrm{Tree}(\sigma )$, which solves the scheme of
problem $s$ relative to the class $C(\alpha )$ and for which, for any
complete path $\tau $ of $S$, the sequence $\beta _{\tau }$ of predicate and
function expressions attached to nodes of $\tau $ coincides with $\beta
_{1},\ldots ,\beta _{m}$. Therefore $\psi (S)=\psi (s)$. Thus, the $\psi $%
-complexity of a scheme of computation tree that is optimal relative to $\psi $, $s$,
and $C(\alpha )$ is at most $\psi (s)$.
\end{proof}

Let $S_{1}=(n,G_{1})$ and $S_{2}=(n,G_{2})$ be schemes of computation trees from the
set $\mathrm{Tree}(\sigma )$. We will say that these schemes are \emph{%
equivalent }if, for any structure $U$ of the signature $\sigma $, the
functions implemented by the computation trees $(S_{1},U)$ and $(S_{2},U)$ coincide.
\index{Scheme of computation tree of signature!equivalent schemes}

\begin{lemma}
\label{L5.3} Any scheme of computation tree $S_{1}=(n,G_{1})\in \mathrm{Tree}(\sigma
)$ can be transformed by changing of variables into a scheme of computation tree $%
S_{2}=(n,G_{2})\in \mathrm{Tree}(\sigma )$, which is equivalent to $S_{1}$
and in which all variables in function and predicate expressions belong to
the set $X_{n+2^{h(S_{1})}}$.
\end{lemma}

\begin{proof}
One can show that the number of function nodes in the scheme $S_{1}$ is at
most $2^{h(S_{1})}$. Each function node $v$ of the scheme $S_{1}$ is labeled
with a function expression $$x_{i(v)}\Leftarrow q_{j(v)}(x_{l(1,v)},\ldots
,x_{l(m(v),v)}),$$ where $m(v)$ is the arity of the function symbol $q_{j(v)}$%
. Let $v_{1},\ldots ,v_{k}$ be all function nodes of the scheme $S_{1}$ for
which variables $x_{i(v_{1})},\ldots ,x_{i(v_{k})}$ do not belong to the set
$X_{n}$. Let $x_{j_{1}},\ldots ,x_{j_{p}}$ be all pairwise different
variables in the sequence $x_{i(v_{1})},\ldots ,x_{i(v_{k})}$. Denote $%
Y=\{x_{j_{1}},\ldots ,x_{j_{p}}\}$.

In all expressions attached to function and predicate nodes of the scheme $%
S_{1}$, we replace each variable that does not belong to the set $X_{n}\cup
Y $ with the variable $x_{n-1}$ and replace variables $x_{j_{1}},\ldots
,x_{j_{p}}$ with variables $x_{n},\ldots ,x_{n+p-1}$, respectively. Denote
by $S_{2}=(n,G_{2})$ the obtained scheme of computation tree. One can show that the
scheme $S_{2}$ is equivalent to the scheme $S_{1}$ and all variables in
function and predicate expressions of the scheme $S_{2}$ belong to the set $%
X_{n+2^{h(S_{1})}}$.
\end{proof}

Let $\psi $ be a strictly limited complexity measure over the signature $%
\sigma $. For $i\in \omega $, we denote $\omega _{\psi
}(i)=\{q_{j}:q_{j}\in \sigma ,\psi (q_{j})=i\}$. Define a partial
function $K_{\psi }:\omega \rightarrow \omega $ as follows. Let $i\in
\omega $. If $\omega _{\psi }(i)$ is a finite set, then $K_{\psi
}(i)=|\omega _{\psi }(i)|$. If $\omega _{\psi }(i)$ is an infinite
set, then the value $K_{\psi }(i)$ is indefinite.

For us the most interesting situation is when the function $K_{\psi }$ is a
total recursive function. If $\sigma $ is a finite signature, then,
evidently, $K_{\psi }$ is a total recursive function.

We now consider a class of strictly limited complexity measures $\psi $ over
infinite  signature $\sigma $ for which $K_{\psi }$ is a total recursive
function. Let $\sigma =\{q_{j}:j\in \omega \}$,  $g:\omega \rightarrow
\omega \setminus \{0\}$ be a nondecreasing unbounded total recursive
function, and $\psi $ be a weighted depth over signature $\sigma $ for which
$\psi (q_{j})=g(j)$ for any $j\in \omega $. Then $\psi $ is a strictly
limited complexity measure over the signature $\sigma $ for which $K_{\psi }$
is a total recursive function.

\begin{theorem}
\label{T5.5}Let $\psi $ be a strictly limited complexity measure over the
signature $\sigma $ for which $K_{\psi }$ is a total recursive function, $C$
be a nonempty class of structures of the signature $\sigma $, $H$ be a
nonempty subset of the set $H(\sigma )$, and the problem of solvability for
the quadruple $(\mathrm{Probl}(\sigma ),\mathrm{Tree}(\sigma ),H,C)$ be
decidable. Then the problem of optimization for the triple $(\psi ,H,C)$ is
decidable.
\end{theorem}

\begin{proof}
Taking into account that the function $K_{\psi }$ is a total recursive
function, it is not difficult to show that there exists an algorithm
constructing the set $\{q_{j}:q_{j}\in \sigma ,\psi (q_{j})\leq r\}$ for
any number $r\in \omega $. From this fact, it follows that there exists an
algorithm, which, for an arbitrary number $r\in \omega $, an arbitrary
number $n\in \omega \setminus \{0\}$, and an arbitrary finite nonempty
subset $M$ of the set $\omega $, constructs the set $\mathrm{Tree}(\sigma
,r,n,M)$ of schemes of computation trees $S=(n,G)$ from $\mathrm{Tree}(\sigma )$
satisfying the following conditions:

\begin{itemize}
\item Terminal nodes of $S$ are labeled with numbers from $M$.

\item $h(S)\leq r$.

\item $\max \{\psi (q_{j}):q_{j}\in \sigma (S)\}\leq r$.

\item All variables in function and predicate expressions in the scheme $S$
belong to the set $X_{n+2^{r}}$.
\end{itemize}

Let $s=(n,\nu ,\beta _{1},\ldots ,\beta _{m})$ be a scheme of problem from $%
\mathrm{Probl}(\sigma )$ and let $M(s)$ be the set of values of the map $\nu
$. Denote $\mathrm{Tree}(s)=\mathrm{Tree}(\sigma ,\psi (s),n,M(s))$. We
now show that the set $\mathrm{Tree}(s)$ contains a scheme of computation tree that
is optimal relative to $\psi $, $s$, and $C(\alpha )$.

Using Lemma \ref{L5.3}, one can show that there exists a scheme of computation tree $%
S\in \mathrm{Tree}(\sigma )$, which is optimal relative to $\psi $, $s$,
and $C(\alpha )$, and in which numbers attached to terminal nodes belong to
the set $M(s)$ and all variables in function and predicate expressions in
the scheme $S$ belong to the set $X_{n+2^{h(S)}}$. Using Lemma \ref{L5.2},
we obtain that $\psi (S)\leq \psi (s)$. From this inequality and Lemma %
\ref{L5.1} it follows that $h(S)\leq \psi (s)$ and $\max \{\psi
(q_{j}):q_{j}\in \sigma (S)\}\leq \psi (s)$. Therefore the scheme $S$, which
is optimal relative to $\psi $, $s$, and $C(\alpha )$, belongs to the set $%
\mathrm{Tree}(s)$.

We now describe an algorithm solving the problem of optimization for the
triple $(\psi ,H,C)$. Let $\alpha \in H$ and $s=(n,\nu ,\beta _{1},\ldots
,\beta _{m})$ be a scheme of problem from the set $\mathrm{Probl}(\sigma )$.
First, we compute the value $\psi (s)$ and construct the set $M(s)$. Next we
construct the set $\mathrm{Tree}(s)=\mathrm{Tree}(\sigma ,\psi (s),n,M(s))$%
. Using the algorithm solving the problem of solvability for the quadruple $(%
\mathrm{Probl}(\sigma ),\mathrm{Tree}(\sigma ),H,C)$, we can find a scheme
of computation tree $S\in \mathrm{Tree}(s)$, which solves the scheme of problem $s$
relative to the class $C(\alpha )$ and has the minimum $\psi $-complexity
among such schemes of computation trees. The scheme of computation tree $S$ is optimal
relative to $\psi $, $s$, and $C(\alpha )$.
\end{proof}

\begin{corollary}
\label{C5.1}Let $\sigma $ be a finite signature, $\psi $ be a strictly
limited complexity measure over the signature $\sigma $, $C$ be a nonempty
class of structures of the signature $\sigma $, $H$ be a nonempty subset of
the set $H(\sigma )$, and the problem of solvability for the quadruple $(%
\mathrm{Probl}(\sigma ),\mathrm{Tree}(\sigma ),H,C)$ be decidable. Then the
problem of optimization for the triple $(\psi ,H,C)$ is decidable.
\end{corollary}

\begin{theorem}
\label{T5.6}Let $C$ be a nonempty class of structures of the signature $%
\sigma $, $H$ be a nonempty subset of the set $H(\sigma )$, $\psi $ be a
strictly limited complexity measure over the signature $\sigma $, and the
problem of satisfiability for the pair $(H\wedge \Phi (P(\exists ^{\ast
}),\sigma ),C)$ be undecidable. Then the problem of optimization for the
triple $(\psi ,H,C)$ is undecidable.
\end{theorem}

\begin{proof}
Let us assume the contrary: the problem of optimization for the triple $%
(\psi ,H,C)$ is decidable. We now describe an algorithm for solving the
problem of satisfiability for the pair $(H\wedge \Phi (P(\exists ^{\ast
}),\sigma ),C)$.

Let $\alpha \in H$ and $\gamma \in \Phi (P(\exists ^{\ast }),\sigma )$. We
construct a sentence of the signature $\sigma $ that is logically equivalent
to the sentence $\gamma $ and is of the form $\exists x_{0}\ldots \exists
x_{n-1}((\gamma _{11}^{\delta _{11}}\wedge \cdots \wedge \gamma
_{1m_{1}}^{\delta _{1m_{1}}})\vee \cdots \vee (\gamma _{k1}^{\delta
_{k1}}\wedge \cdots \wedge \gamma _{km_{k}}^{\delta _{km_{k}}}))$, where,
for $j=1,\ldots ,k$ and $i=1,\ldots ,m_{j}$, $\delta _{ji}\in E_{2}$ and $%
\gamma _{ji}\ $is an atomic formula of the signature $\sigma $ of the form $%
q_{l}(t_{1},\ldots ,t_{p})$, where $q_{l}$ is a $p$-ary predicate symbol
from $\sigma $ and $t_{1},\ldots ,t_{p}$ are terms of the signature $\sigma $
with variables from $X_{n}$.

For $j=1,\ldots ,k$, we construct a scheme of problem $s_{j}$ from $\mathrm{%
Probl}(\sigma )$ with special representation $(n,\nu _{j},\gamma
_{j1},\ldots ,\gamma _{jm_{j}})$ such that $\nu
_{j}:E_{2}^{m_{j}}\rightarrow E_{2}$ and, for any $\bar{\delta}\in
E_{2}^{m_{j}}$, $\nu _{j}(\bar{\delta})=1$ if and only if $\bar{\delta}%
=(\delta _{j1},\ldots ,\delta _{jm_{j}})$. It is clear that, for any
structure $U$ of the signature $\sigma $ and any tuple $\bar{a}\in A^{n}$,
where $A$ is the universe of $U$,%
\[
\varphi _{(s_{j},U)}(\bar{a})=\left\{
\begin{array}{ll}
0, & \mathrm{if}\ U\models \lnot (\gamma _{j1}(\bar{a})^{\delta _{j1}}\wedge
\cdots \wedge \gamma _{jm_{j}}(\bar{a})^{\delta _{jm_{j}}}), \\
1, & \mathrm{if}\ U\models \gamma _{j1}(\bar{a})^{\delta _{j1}}\wedge \cdots
\wedge \gamma _{jm_{j}}(\bar{a})^{\delta _{jm_{j}}}.%
\end{array}%
\right.
\]

We denote by $S_{0}$ the scheme of computation tree from $\mathrm{Tree}(\sigma )$,
which consists of only one node labeled with the number $0$. It is clear
that, for any structure $U$ of the signature $\sigma $ and any tuple $\bar{a}%
\in A^{n}$, where $A$ is the universe of $U$, $\varphi _{(S_{0},U)}(\bar{a}%
)=0$.

Using an algorithm that solves the problem of optimization for the triple $%
(\psi ,H,C)$, for $j=1,\ldots ,k$, we find a scheme of computation tree $S_{j}\in $ $%
\mathrm{Tree}(\sigma )$, which solves the scheme of problem $s_{j}$
relative to the class $C(\alpha )$ and has the minimum $\psi $-complexity.

Using the properties of the strictly limited complexity measure $\psi $, it
is not difficult to show that a structure $U\in $ $C$ such that $U\models
\alpha \wedge \gamma $ exists if and only if there exists $j\in \{1,\ldots
,k\}$ for which the scheme of computation tree $S_{j}$ does not coincide with the
scheme of computation tree $S_{0}$. Therefore the problem of satisfiability for the
pair $(H\wedge \Phi (P(\exists ^{\ast }),\sigma ),C)$ is decidable but this
is impossible.
\end{proof}

\subsection{Equality Is Allowed}

We now consider the case, when the equality is allowed.

Let $\sigma $ be a finite or countable signature. If $\sigma $ is finite,
then we represent it in the form $\sigma =\{q_{1},\ldots ,q_{m}\}$. If $%
\sigma $ is infinite, then we represent it in the form $\sigma
=\{q_{1},q_{2},\ldots \}$. Let $\sigma _{=}=\sigma \cup \{q_{0}\}$, where $%
q_{0}$ is the symbol denoting equality $=$, and $\sigma _{=}^{\ast }$ be the
set of finite words in the alphabet $\sigma _{=}$ including the empty word $%
\lambda $.

\begin{definition}
A \emph{e-complexity measure} over the signature $\sigma $ is an arbitrary
map $\psi :\sigma _{=}^{\ast }\rightarrow \omega $. The e-complexity measure
$\psi $ will be called \emph{strictly limited} if it is computable and,
for any $\alpha _{1},\alpha _{2},\alpha _{3}\in \sigma _{=}^{\ast }$, it
possesses the following property: if $\alpha _{2}\neq \lambda $, then $\psi (\alpha _{1}\alpha
_{2}\alpha _{3})>\psi (\alpha _{1}\alpha _{3})$.
\end{definition}
\index{E-complexity measure over signature}

The prefix \textquotedblleft e-\textquotedblright\ here and later indicates
the presence of the equality.

Let $\psi $ be an e-complexity measure over the signature $\sigma $. We
extend it to the sets $\mathrm{Probl}^{=}(\sigma )\ $and $\mathrm{Tree}%
^{=}(\sigma )$. Let $\beta $ be a finite sequence of function and predicate
expressions of the signature $\sigma $. We correspond to $\beta $ a word $%
\mathrm{word}(\beta )$ from $\sigma _{=}^{\ast }$. If the length of $\beta $
is equal to $0$, then $\mathrm{word}(\beta )=\lambda $. If $\beta =\beta
_{1},\ldots ,\beta _{m}$, then $\mathrm{word}(\beta )=$ $q_{i_{1}}\cdots
q_{i_{m}}$, where, for $j=1,\ldots ,m$, $q_{i_{j}}$ is the symbol from $%
\sigma _{=}$ contained in the expression $\beta _{j}$. In particular, if $%
\beta _{j}$ has the form $x_{l_{1}}=x_{l_{2}}$, then $q_{i_{j}}=q_{0}$.

Let $s=(n,\nu ,\beta _{1},\ldots ,\beta _{m})$ be a scheme of problem from
the set $\mathrm{Probl}^{=}(\sigma )$. Then $\psi (s)=\psi (\mathrm{word}%
(\beta _{1},\ldots ,\beta _{m}))$.

Let $S=(n,G)$ be a scheme of computation tree from the set $\mathrm{Tree}^{=}(\sigma
)$ and $\tau =w_{1},d_{1},w_{2},$ $d_{2},\ldots ,w_{m},d_{m},w_{m+1}$ be a
complete path of the scheme $S$. We denote $\beta _{\tau }=\beta _{1},\ldots
,\beta _{m}$ the sequence of function and predicate expressions attached to
nodes $w_{1},\ldots ,w_{m}$. Then $\psi (S)=\max \{\psi (\mathrm{word}(\beta
_{\tau })):\tau \in \mathrm{Path}(S)\}$, where $\mathrm{Path}(S)$ is the set
of complete paths of the scheme $S$. The value $\psi (S)$ is called the $%
\psi $-\emph{complexity} of the scheme of computation tree $S$.

Let us remind that $H^{=}(\sigma )$ is the set of sentences of the signature
$\sigma $. Let $\psi $ be a strictly limited e-complexity measure over the
signature $\sigma $, $C$ be a nonempty class of structures of the signature $%
\sigma $, and $H$ be a nonempty subset of the set $H^{=}(\sigma )$.

\begin{definition}
We now define the \emph{problem of e-optimization} for the triple $(\psi
,H,C)$: for arbitrary sentence $\alpha \in H$ and scheme of problem $s\in
\mathrm{Probl}^{=}(\sigma )$, we should find a scheme of computation tree $S\in
\mathrm{Tree}^{=}(\sigma )$, which solves the scheme of problem $s$
relative to the class $C(\alpha )$ and has the minimum $\psi $-complexity.
\end{definition}
\index{Algorithmic problem!problem of e-optimization}

Let $\psi $ be a strictly limited e-complexity measure over the signature $%
\sigma $. For $i\in \omega $, we denote $\omega _{\psi
}(i)=\{q_{j}:q_{j}\in \sigma _{=},\psi (q_{j})=i\}$. Define a partial
function $K_{\psi }:\omega \rightarrow \omega $ as follows. Let $i\in
\omega $. If $\omega _{\psi }(i)$ is a finite set, then $K_{\psi
}(i)=|\omega _{\psi }(i)|$. If $\omega _{\psi }(i)$ is an infinite
set, then the value $K_{\psi }(i)$ is indefinite.

Proof of the next statement is similar to the proof of Theorem \ref{T5.5}.

\begin{theorem}
\label{T5.7}Let $\psi $ be a strictly limited e-complexity measure over the
signature $\sigma $ for which $K_{\psi }$ is a total recursive function,  $C$
be a nonempty class of structures of the signature $\sigma $, $H$ be a
nonempty subset of the set $H^{=}(\sigma )$, and the problem of solvability
for the quadruple $(\mathrm{Probl}^{=}(\sigma ),\mathrm{Tree}^{=}(\sigma
),H,C)$ be decidable. Then the problem of e-optimization for the triple $%
(\psi ,H,C)$ is decidable.
\end{theorem}

\begin{corollary}
\label{C5.2}Let $\sigma $ be a finite signature, $\psi $ be a strictly
limited e-complexity measure over the signature $\sigma $, $C$ be a nonempty
class of structures of the signature $\sigma $, $H$ be a nonempty subset of
the set $H^{=}(\sigma )$, and the problem of solvability for the quadruple $(%
\mathrm{Probl}^{=}(\sigma ),\mathrm{Tree}^{=}(\sigma ),H,C)$ be decidable.
Then the problem of e-optimization for the triple $(\psi ,H,C)$ is decidable.
\end{corollary}

Proof of the next statement is similar to the proof of Theorem \ref{T5.6}.

\begin{theorem}
\label{T5.8}Let $C$ be a nonempty class of structures of the signature $%
\sigma $, $H$ be a nonempty subset of the set $H^{=}(\sigma )$, $\psi $ be a
strictly limited e-complexity measure over the signature $\sigma $, and the
problem of satisfiability for the pair $(H\wedge \Phi ^{=}(P(\exists ^{\ast
}),\sigma ),C)$ be undecidable. Then the problem of e-optimization for the
triple $(\psi ,H,C)$ is undecidable.
\end{theorem}

\section{Conclusions\label{S5.7}}

In this paper, we studied relationships among three algorithmic problems involving computation trees: the optimization, solvability, and satisfiability problems. We also studied the decidability of these problems in different situations. In the future, we are planning to consider heuristics for the optimization of computation trees.

\subsection*{Acknowledgements}

Research reported in this publication was supported by King Abdullah
University of Science and Technology (KAUST).

\bibliographystyle{spmpsci}
\bibliography{abc_bibliography}

\end{document}